\documentclass[11pt,reqno]{amsart}

\usepackage[T1]{fontenc}
\usepackage[utf8]{inputenc}
\usepackage{lmodern}
\usepackage{microtype}
\usepackage{amsmath,amssymb,amsfonts,amscd,mathtools}
\usepackage{latexsym,array,hhline,xcolor,graphicx,stmaryrd}
\usepackage{enumitem}
\usepackage[a4paper,margin=1in]{geometry}
\usepackage[numbers,square,sort&compress]{natbib}
\usepackage[colorlinks=true,citecolor=blue,linkcolor=blue,urlcolor=blue]{hyperref}
\usepackage[nameinlink,noabbrev]{cleveref}

\newcommand{\doi}[1]{\href{https://doi.org/#1}{doi:\nolinkurl{#1}}}
\newcommand{\isbn}[1]{\href{https://search.worldcat.org/isbn/#1}{ISBN~\nolinkurl{#1}}}

\hypersetup{
  pdftitle={On convergence of the Mayer problems arising in the theory of financial markets with transaction cost},
  pdfauthor={Yuri Kabanov and Artur Sidorenko},
  pdfkeywords={Markets with transaction cost, portfolio optimization, Mayer's problem, weak convergence}
}

\allowdisplaybreaks
\theoremstyle{plain}
\newtheorem{theo}{Theorem}[section]

\newtheorem{prop}[theo]{Proposition}

\newtheorem{lemm}[theo]{Lemma}

\newtheorem{coro}[theo]{Corollary}

\theoremstyle{definition}

\theoremstyle{remark}

\newtheorem{remark}[theo]{Remark}

\numberwithin{equation}{section}

\newcommand\cA{\mathcal{A}}
\newcommand\cE{\mathcal{E}}

\newcommand\cF{\mathcal{F}}
\newcommand\cG{\mathcal{G}}

\newcommand\cL{\mathcal{L}}
\newcommand\cB{\mathcal{B}}

\newcommand\cM{\mathcal{M}}
\newcommand\cX{\mathcal{X}}
\newcommand\cD{\mathcal{D}}

\newcommand\cO{\mathcal{O}}
\newcommand\e{{\varepsilon}}

\newcommand{\wh}{\widehat}

\newcommand{\zs}[1]{{\mathchoice{#1}{#1}{\lower.25ex\hbox{$\scriptstyle#1$}}
{\lower0.25ex\hbox{$\scriptscriptstyle#1$}}}}

\newcommand\beq{\begin{equation}}
\newcommand\eeq{\end{equation}}
\newcommand\bea{\begin{align}}
\newcommand\eea{\end{align}}
\newcommand\bean{\begin{align*}}
\newcommand\eean{\end{align*}}
\newcommand\beal{\begin{align}}
\newcommand\eeal{\end{align}}

\def\bbb{{\mathbb B}}
\def\bbr{{\mathbb R}}
\def\N{{\mathbb N}}
\def\bbn{{\mathbb N}}

\def\E{{\bf E}}
\def\P{{\bf P}}
\def\Q{{\bf Q}}

\def\F{{\bf F}}

\def\M{{\bf M}}

\def\R{{\bf R}}

\def\build #1_#2{\mathrel{\mathop{\kern 0pt #1}\limits_{#2}}}
\newcommand\1{{\bf 1}}

\title[On convergence of the Mayer problems]{On convergence of the Mayer problems arising in the theory of financial markets with transaction cost}

\author{Yuri Kabanov}

\author{Artur Sidorenko}


\begin{document}

\begin{abstract}
The geometric approach to financial markets with proportional  transaction cost prescribes to imbed a specific model (of  stock market, of  currency market etc.), usually given in a parametric form,
into a natural framework  defined by the two random processes, $S$ and $K$. The first one, $d$-dimensional, models the price evolution of  basic securities while the second one, cone-valued, describes the evolution of the solvency  set.   It happened that  the fundamental  questions --- no-arbitrage criteria, hedging problems, portfolio optimization --- can be studied in this general setting  opening the door to set-valued techniques. In this note we explore, in such a general framework,  the stochastic Mayer control problem, consisting in the  maximization   of the  expected utility of the portfolio terminal wealth.  We get results on  continuity of the optimal value and the optimal control under price approximations in a general multi-asset framework described by the geometric formalism.
\end{abstract}

\maketitle

\begin{center}
\small
\begin{minipage}{0.88\textwidth}
\noindent\textbf{Yuri Kabanov.} Lomonosov Moscow State University and ``Vega'' Institute Foundation, Moscow, Russia, and Universit\'e de Marie et Louis Pasteur, Laboratoire de Math\'ematiques, UMR CNRS 6623, Besan\c{c}on, France.\\
\texttt{ykabanov@univ-fcomte.fr}

\medskip
\noindent\textbf{Artur Sidorenko.} Lomonosov Moscow State University and ``Vega'' Institute Foundation, Moscow, Russia.\\
\texttt{sidorenkoap@my.msu.ru}
\end{minipage}
\end{center}

\section{Introduction}

This work  is inspired by the paper of Bayraktar, Dolinskyi, and Dolinsky \cite{BDD}  on the  stability of a stochastic Mayer 
problem under modeling errors. The latter consists in   maximization of the expected utility of the terminal value
of a controlled process, usually, a solution to a stochastic equation.  This is a singular control problem that has been studied for more than five decades: the seminal paper by Magill and Constantinides, \cite{MC}, was published in 1976. The theory was placed in a rigorous mathematical framework by Davis and Norton,  \cite{DN}, and by Shreve and Soner \cite{Shreve-Soner} who used the theory of viscosity solutions. It was further developed in various directions by a number of authors; see, for instance, \cite{CK1996,CW2001,Bouchard2002,Guasoni2002,CzichowskySchachermayer2016,CPSY2018}.  A breakthrough in the theory was a note \cite{K} where it was shown  that 
various multi-asset models with friction can be treated within a unified geometric framework using the convex duality in a finite-dimensional space. In that note and  subsequent works a bridge between mathematical economics and mathematical finance was built:   it was recommended to consider simultaneously the  monetary representation of assets and the representation in physical units, which plays a fundamental role in the understanding of arbitrage theory, see the book \cite{KabS} and references therein to works by Kabanov, Stricker, Delbaen, Schachermayer, R\'asonyi and others who contributed to its further development.

The  questions addressed in \cite{BDD} (see also \cite{BCDD2021})   are the  following:   do the optimal values and the
optimal strategies converge when the stock prices converge in law? In other words, is the setting sensitive to errors in the model description? Of course, these questions are of great interest  even for the models without transaction costs, see  \cite{BDG,kreps2020,hubalek2021,BBBE2020}.
 
The paper \cite{BDD} treats the model with only one risky asset. 
It is natural to consider more realistic multi-asset models and, more generally, an abstract mathematical formulation covering various cases and allowing one to reveal the essential features of the problem.

The modern theory of financial markets  provides  a  natural geometric framework allowing  to cover in a uniform way various types of financial markets, e.g., a stock market where  only two actions are allowed  (''buy  stock`` and ''sell  stock``),  a currency market where any currency can be exchanged to any other currency, models where  transactions charge only the bank account, the mentioned model with linear constraints on transactions  etc., see Chapter 3 in the book \cite{KabS}.  

For the reader's convenience we briefly recall  the evolution of the theory and its current state of the art.  At the initial stages,   financial markets with proportional transaction costs were described in  parametric forms: both for the price processes (e.g.,  geometric Brownian motions or geometric L\'evy processes) and for the transaction costs given  by  matrices of transaction costs coefficients or bid-ask spreads etc., see \cite{DN,JS,CK1996}.  Later, 
as we already mentioned, 
it was observed that many typical problems, such as  no-arbitrage conditions, hedging theorems, and portfolio optimization problems,  can be  treated in a more general ``geometric" framework, see, e.g., \cite{K,DPT2001,Bouchard,KL2002,KK2004,DKL2016} and Chapters 3 and 4 in the book \cite{KabS}.  This framework involves, besides the  price process  of basic assets $S=(S_t)$, usually, a semimartingale, also geometric objects --- the  solvency sets forming   an adapted cone-valued process $K=(K_t)$  containing the first orthant. Usually, $K$ is given in monetary units  while the symbol  $\widehat K$ stands to express solvency  in terms of physical units. Further developments of the geometric approach include, for instance, discontinuous prices \cite{CS2006,CO2011} and cumulative prospect theory preferences \cite{CR2020,Chau-Rasonyi2017}.   Related ideas appear in the recent framework based on $\varepsilon$-arbitrage \cite{Acciaio2025}. 

In the important case, corresponding to the transaction costs not evolving in time,  $K$ is just a constant convex cone (the word convex  usually  omitted in the considered context). 
In particular, for models of markets without friction the cone $K=\{x\in \bbr^d\colon {\bf 1}x\ge 0\}$ where ${\bf 1}:=(1,...,1)$. In other words, $K$ is a half-space consisting of vectors whose coordinates sum to a nonnegative number.

If all transactions are charged, the solvency cone $K$ is proper, that is, the intersection $K \cap (-K) = \{ 0 \}$. This  property, referred to as efficient friction,  significantly simplifies the analysis.

Though the solvency cones in financial market models are polyhedral, from the point of view of the optimal control theory, this is not a common constraint. A question arises  whether it can be avoided, at least, without heavy mathematical sophistication. 

The theory  also recommends to work jointly with  the monetary representation of portfolio positions and their representation in physical units.  The first one involves price evolution and  investor's actions (i.e. the portfolio revision).  The representation in terms of physical units has  much simpler dynamics determined only by investor's actions.  In both cases, the dynamics are vector-valued and investor's utility may depend on the whole vector of  positions in various assets. 
Note that the solvency cones in the physical units representation are random set-valued processes, even  if  $K$ is constant in the monetary representation. It is worth mentioning that the two-asset model (i.e., with only one risky asset) is very specific because the only convex cones on the plane are sectors. The multi-asset theory provides a lot of new effects, see Subsection  3.2.3 in \cite{KabS}  for examples of results that cannot  be directly extended to the multi-asset framework.   

We distinguish here the model, the pair $(\cL(Y,S),K)$ and a model realization ${\bf M}={\bf M}(\bbb,Y,S,K)$ where $\bbb=(\Omega, {\cF}, \F=({\cF}_t)_{t \in [0, T]}, \P)$ is a stochastic basis. The information flow, i.e. the filtration ${\bf F^Y}$,  is  generated  by a right-continuous process $Y$ with independent increments (one can consider more general information flow but this makes theory too technical). 
The initial value $Y_0$ can  be arbitrary random variable independent on the increments of $Y$. The $d$-dimensional  price process $S$ is assumed to be  continuous and the joint law of $(Y,S)$ is denoted by $\cL(Y,S)$ and is given in the model. 
 It is important to note that the filtrations generated by $S$ and  $Y$ are not specified and may vary  from one realization to another.
The constant solvency cone $K$ corresponds to  the monetary representation. We define  the set $\cA(x)$ of admissible strategies  where $x\in \bbr$ is an initial value. The admissibility of a strategy means that the 
controlled process, given by a linear dynamics,  does not leave the solvency cone. 
For the detailed description see Section \ref{problem}.  According to a tradition,  we call {\it stochastic  Mayer problem} for ${\bf M}$
the problem of maximization of the terminal wealth. In our case, we 
maximize, over the strategies from $\cA(x)$, the  expected  utility from the  terminal values of the controlled process. 

The question  of interest is the sensitivity (or stability) of the model with respect to the factors. In our case it is the  continuity property  of the Bellman function $u(x, \M)$ with respect to a joint distribution of price process $S$ and process $Y$ generating the underlying filtration. Does the Bellman function $u(x, \M)$ depend only on the measures $\cL(Y,S)$ or on the realization?  
Does a weak  convergence  $\cL(Y^n, S^n)$ to  $\cL(Y, S)$ imply the convergence of $u(x,\M^n)$ to $u(x,\M)$? 
If not, what kind of convergence  is needed to ensure the latter property? 
What  can one say about the convergence of optimal strategies? 

The  above  description of the model does not define it in a unique way and this circumstance seems to be often overlooked. To explain the phenomenon let us consider the simplest case where $S$ is just a Wiener process $W$ generating the filtration ${\bf F}^W$. Let  
$a(s,W)$ be the Tsirelson functional \cite{Tsirelson1975} of the famous  example  where a strong solution of the SDE with drift $a$ depending on the trajectory does not exist. According to the Girsanov theorem the process $\tilde W_t=W_t-\int_0^t a(s,W)ds$ is a Wiener process with respect to the equivalent probability measure $\tilde P$ with the density  given by the stochastic exponential,  that is with  $d\tilde \P/d\P:=\cE_T(a\cdot W)$.  The filtration ${\bf F}^{\tilde W}$ is strictly smaller than  ${\bf F}^W$.  Apparently, without assumptions on the utility function, the optimal value on the basis 
$(\Omega, {\cF}, {\bf F}^{\tilde W}=(\cF^{\tilde W}_t)_{t \in [0, T]}, \tilde \P)$ may be smaller than that on the initial one. 
Note also that we can apply the Girsanov theorem with $\tilde W$ and get another Wiener process with a strictly smaller filtration and so on.  
Also, the initial extension of the filtration may change the optimal solution of the problem, see \cite{CR2015} for a discrete-time example without friction. This example is based on the Choquet integral over a disturbed probability measure. 


Our contribution is twofold.
First, we establish stability of the portfolio optimization problem in a multi-asset market with proportional transaction costs under weak convergence of the underlying price and factor processes. 
Second, we distinguish the market model as a probability distribution of price and factors and the model realization on a particular stochastic basis. We show that the Bellman function of the utility maximization problem in our setting depends only on the model and does not depend on its realization on a particular stochastic basis.

The paper is organized as follows. In Section \ref{problem}, we provide a detailed description of our setup. In Section \ref{section:skorokhod}, we establish  the independence of the Bellman function of a particular realization of $(S,Y)$.
In Section \ref{sec:main_result},  the main continuity theorem is formulated. Further sections contain proofs of this theorem. 
The proof is split in  three  steps.  In Sections  \ref{sec:lowe} and \ref{sec:upper}, we verify, respectively,  the lower semi-continuity  and upper-semicontinuity properties of the Bellman function and complete the proof of the main result. Auxiliary statements are relegated to Appendix.

\section{Basic concepts}
\label{problem}

The  practical  modeling of any financial market  is based on the observable prices which are causal functions of observable or unobservable  factors randomly evolving and  influencing the  price dynamics.  The transaction costs parameters (charges of the market, taxes etc.) in the majority of cases are given exogenously.   In the simplest case of constant proportional transaction costs they are described by a matrix defining  the solvency cone $K$, which is,  in the geometric approach, the primary object, see Section 3.1 in the book \cite{KabS}.   

We  formalize this  situation in accordance with the  classical probabilistic set-up  as follows.  There is the ``Universe``,  a filtered probability  space, or stochastic basis,  $\bbb:=(\Omega, {\cF}, {\bf F}=(\cF_t)_{t\le T}, \P)$.   The price process $S$ with values in $\bbr^d$, $d\ge 2$, is given by its law and is adapted to a  subfiltration ${\bf F}^Y$ generated by a process $Y$, also given by its law and  describing the unobserved factors. Usually we take ${\bf F}$ equal to ${\bf F}^Y$.

{\sl Example 1.} The Black--Scholes model, ``textbook version'', $S=(S^0,S^1)$.  The bank account is the numeraire, i.e. $S^0=1$.   The price process of risky asset has the law $\cL$ of   geometric Brownian motion with constant parameters $(a, \sigma)$.   
We can ``realize'' the model and define  a stochastic basis with the price process   $S^1=\cE (at+\sigma W)$ where $W$ is a Wiener process;    we can take  $Y=W$ (or  $Y=\xi+W$ where $\xi$ is a random variable independent of $W$). As we already noticed, even in this simple case, the filtrations generated by the Wiener process can vary.

{\sl Example 2.}  The law $\cL$ of the price of the risky asset is the law of conditional geometric Brownian motion with the parameters switched according to a telegraph process, i.e. by a Markov process with two states describing external factors.  Such a model  reflects the existence of two regimes of the environment which can  describe business cycles of the economy. The process modulating  parameters can be represented as a function of a process with independent increments.   

We call a {\it model} the pair   $(\cL(Y,S), K)$ where $\cL(Y,S)$ is a probability distribution of an $m+d$-dimensional  process (i.e., a probability measure in the Skorokhod  space of c\`adl\`ag functions). In this paper we restrict ourselves to the case where  $\cL(S)$  is a probability distribution in the space $C^d[0,T]$ of continuous functions.  The last symbol $K$ denotes  a closed (convex) cone in $\bbr^d$ such that ${\rm int}\, K \supset\bbr^d_+ \setminus  \{0\}$.  
 
We call a {\it realization of the model} the quadruple  $\M=(\bbb, Y,S,K)$ where the law $\cL(Y,S)$ and $K$ satisfy the properties formulated  above and $S$ is measurable with respect to the filtration ${\bf F}^Y$.  We shall call realizations $\M$  and $\M'$ equivalent (in symbols: $\M\sim \M')$  if they originate from the same model. 

\begin{remark} Not every model admits a realization. For example, if $\cL(Y)$ and $\cL(S)$ are the laws of nontrivial processes with independent increments and  $\cL(Y,S)=\cL(Y)\otimes \cL(S)$,  the model does not admit a realization: $S$ cannot be adapted with the filtration ${\bf F}^Y$
generated by $Y$. We do not elaborate further on the suggested formalism because,  in the majority of  cases, the description of a model starts with its particular  realization (``We are given a stochastic basis...'').  We want only to emphasize that the  choice of the latter may influence  the desired result.  As was already explained,  
the Wiener process may have realization with different filtration.  A difference in the information flows may lead to  a different Bellman function in optimal control problems.      
\end{remark} 

In the financial context $S$ describes the price evolution  of $d\ge 2$ basic assets measured in some accounting units. Without loss of generality, we assume, as in the above examples, that $S^1\equiv 1$, that is the first traded asset is chosen as the  numeraire.  The process $S$ can be a strong solution to an SDE.  The cone $K$ is interpreted as the {\it solvency region}. Typically, e.g.,  in  models of  currency markets or stock markets, it is given by matrices of transaction costs coefficients. For example,  by 
the matrix $\Lambda=(\lambda^{ij})$ where $\lambda^{ij}\ge 0$, $\lambda^{ii}=0$, whose  entries  have the following  interpretation:  the one-unit increase  of the position   $i$ requires the  transfer of $1+\lambda^{ij}$  from the position  $j$.  With this parameterization 
\[
\begin{aligned}
K=\Big \{x\in \bbr^d\colon\;&\hbox{there is a matrix $(a^{ij})$ with  $a^{ij}\ge 0$}\\
&\hbox{such that }\ x^i+\sum_{j}(a^{ij}-(1+\lambda ^{ji})a^{ji})\ge 0\ \forall i \Big\}.
\end{aligned}
\]

In the present paper  we assume that  the transaction costs are constant over time and so is the cone  $K$ but in a more general setting the solvency cone may depend on time and even be a cone-valued random process. 
The reader  may always think that $K$ is a polyhedral cone as in the above example (and  in all existing financial models).  Our setting includes a market model without friction where $K=\{x\in \bbr^d\colon {\bf 1}x\ge 0\}$, $ {\bf 1}=(1,\dots,1)$.
   
\smallskip  
Let us fix a model realization and introduce the basic concepts. 

We associate with the cone $K$ and the process $S$ the cone-valued random process  $\widehat K:=K/S$.  This intuitive symbolical notation  means that for each $t$ (and $\omega$ hidden as always) $\widehat K_t:=\varphi_t K$, where 
\[
\varphi_t: (x^1,...,x^d)\mapsto (x^1/S^1_t,\dots,x^d/S^d_t), \quad t\in [0,T]. 
\]

In the theory, the dual cone-valued process $\widehat K^*$ with $\widehat K^*_t := \varphi^{-1}_t K^*$, where $K^* := \{w \in \bbr^d : wx \ge 0\ \forall\, x \in K\} \subset \bbr^d_+$ is the (positive) dual of the cone $K$, plays an important role.

\smallskip
A function $f: \, [0, T] \to \bbr^d$ is $K$-{\it decreasing} if $f_s - f_t \in K$ when $s\le t$. The definition of $K$-increasing function is similar. 

\smallskip

Any $\bbr^d$-valued  ${\bf F}$-adapted c\`adl\`ag $K$-decreasing process of bounded variation will be called {\it control} or {\it strategy}. For a control  $B$ and  $x \in K$ we define the process $\widehat V=\widehat V^{x,B}$ with  the components $\widehat V^i=x^i+ (1/S^i)\cdot B^i$ and the process 
$V=\varphi^{-1}\widehat V$ with  
\[
V^i=S^i\widehat V^i=S^i(x^i+ (1/S^i)\cdot 
B^i).
\] 
If $S$ is a semimartingale, the product formula implies that $V^i=x^i+V^i\cdot L^ i+B^i$
where $L^i :=  (1/S^i) \cdot S^i$. Alternatively, $V^i=x^i+\widehat V^ i\cdot S^ i+B^i$. Here and in the sequel we use the standard notation ``$\cdot$'' for the integrals. 

The processes $\widehat V^{x,B}$ and $V^{x,B}$  in the context of specific financial models describe the evolutions of investor's positions in each of $d$ assets where $x$ is  the initial value $x$ and $B$ is the control. Consistently with the introduced notations,  $\widehat V^{x,B}$ presents the dynamics  in terms of the  physical units and $V^{x,B}$ gives the description   in monetary value, i.e. in terms of the numeraire. 

Define 
the set $\cA(x)$ of {\it admissible strategies} as the set of the controls $B$ such that  
$\widehat V^{x,B}\in \widehat K$ (up to a negligible set),  that is,  $\widehat V^{x,B}_t\in {\widehat K_t}$ (a.s.)
for all $t\in [0,T]$. It is easy to see that  $\cA(x)$ is a convex set. The sets $\cA(x)$ have the following properties: $\cA(y)\supseteq \cA(x)$ if $y-x\in K$,  $\cA(\lambda x)=\lambda \cA(x)$  $\forall\,\lambda>0$, and  
\[
\alpha \cA(x) +(1-\alpha)\cA(y)\subseteq \cA(\alpha x+ (1-\alpha)y),  \quad \forall\, \alpha\in [0,1].
\]

In models with transaction costs it seems more natural to consider the utility of the terminal value of the portfolio process in physical units rather than in terms of the numeraire \cite{CS2006,CO2011}. 
With this remark we formulate  the optimal control problem as follows:  to maximize over $\cA(x)$ the expected utility   $\E [U(V^{x,B}_T, S)]$, i.e. to find   
the Bellman function 
\beq
\label{u}
u(x,\M):= \sup_{B \in \cA(x)}\E [U(V^{x,B}_T, S)] 
\eeq
where the utility function $U:K \times C^d[0,T]\to \bbr_+$ is concave  and increasing in the first argument with respect to the componentwise partial ordering in $\bbr^d$  and continuous in the second argument. Note that we restrict ourselves to positive utility functions though  extensions for more general case are possible.
Further assumptions  will be introduced later.

\section{Structural properties of the Bellman function}
 \label{section:skorokhod}

In this section, we study properties of the Bellman function with respect to the model realization.
We consider the following problem. Let $\M$ be a realization of a model. Without loss of generality, one may assume that the probability space $\Omega$ supports a random variable $\xi$ independent of $\cF^{Y}_T$ and  uniformly distributed in $[0, 1]$. 
We define a ``randomized'' realization $\tilde \M := (\tilde \bbb, Y,S,K)$ and  $\tilde \bbb := (\Omega, \cF, \tilde {\bf F}, \P)$, where  $\tilde {\bf  F} = (\tilde \cF_t)_{t \le  T}$ with $\tilde \cF_t = \sigma \{  \cF_t, \xi \}$. The question of interest is whether the Bellman functions satisfy $u(x, \M) = u(x, \tilde \M)$.

Following \cite{Chau-Rasonyi2017}, we put 
\[
Y^{t} := Y \1_{\llbracket 0, t \llbracket} + Y_t \1_{\llbracket t, \infty \llbracket}  
\]
and
\[
_t \! Y := (Y- Y_t) \1_{\llbracket t, \infty \llbracket}.  
\]

\begin{theo}
\label{theo:rand}
For every $x \in {\rm int} \, K$
\[
u(x, {\bf M}) = u(x, {\bf \tilde M}).  
\]
\end{theo}

\begin{proof}
    The  inequality $u(x, {\bf M}) \leq u(x, {\bf \tilde M})$ is immediate, so it remains  to prove the converse inequality.
    Let $B \in \cA(x, {\bf \tilde M})$. By the Doob theorem, there exists a Borel function $f: \cD^m_T \times [0, 1] \mapsto \cD^d_T$ such that  the process $B = f(Y, \xi)$. Our goal is to verify that for a.a. $a \in [0, 1]$, the process $\bar B(a) := f(Y, a) \in \cA(x, {\bf M})$. 
    
    Fix $t \in [0, T]$. It is easily seen that $B_t = \pi_t(f(Y, \xi))$ where $\pi_t(x) = x_t$. On the other hand, $B_t$ is $\tilde \cF_t$-measurable implying $\pi_t(f(Y, \xi)) = f_t(Y^t, \xi)$ for some Borel function $f_t$. By independence of $Y$ and $\xi$, and by the Fubini theorem, $f_t(Y^t, a) = \pi_t(f(Y, a))$ for a.a. $a \in [0, 1]$. It follows that for a dense countable set $I \subset [0, T]$ with $\{0, T \} \in I$, r.v. $\bar B_t$ are $\cF_t$-measurable for a.a. $a \in [0, 1]$. Since $B_t = \lim_{I \ni s \to t+} B_s$ for $t \neq T$ and ${\bf F}$ is right-continuous, $\bar B(a)$ is adapted for a.a. $a \in [0, 1]$. On the other hand, 
    \[
    \1_{\wh V^{x,B}_t \in \wh K_t, \, \forall t \in [0, T] } = g(Y, \xi)
    \]
    for some Borel function $g$. Immediately, $g(Y, \xi) = 1$ a.s. Again, due to the Fubini theorem, $g(Y, a) = 1$ for a.a. $a \in [0, 1]$. Thus, we have established that $\bar B(a) \in \cA(x, {\bf M})$ for a.a. $a \in [0, 1]$.

    Finally, fix $\varepsilon > 0$. Then, for some strategy $B \in \cA(x, {\bf \tilde M})$,
    \begin{multline*}
    u(x, \tilde {\bf M}) - \varepsilon \leq \E U(V^{x, B}_T, S) = \E \left[ \E \left[ U(V^{x, B}_T, S) \, | \, \xi \right] \right] = \\ 
     \int_{[0, 1]} \E  U(V^{x, \bar B(a)}_T, S)  \, da \leq \int_{[0, 1]} u(x, {\bf M}) da = u(x, {\bf M}).
    \end{multline*}
    As $\varepsilon$ is arbitrary, the assertion follows.
\end{proof}

\begin{remark}
    \label{rem:random_max}
    If $B^\dagger \in \cA(x, {\bf \tilde M})$ attains the maximal value $u(x, {\bf \tilde M})$ then the process $\bar B^\dagger(a)$ attains the maximal value $u(x, {\bf M})$ for almost all $a \in [0, 1]$.
\end{remark}

Now we are interested in the following question: do the Bellman functions of two equivalent model realization $\M_1$ and  $\M_2$ coincide? In our framework the answer is positive. To prove this statement, we introduce model realizations with initial randomization of the filtration. 

\begin{theo}
\label{theo:indep}
Consider two model realizations ${\bf M}^1 \sim {\bf M}^2$. Then the following holds for all $x \in K$.
\beq
u(x, {\bf M}^1) = u(x, {\bf M}^2).
\eeq
\end{theo}
\begin{proof}
Thanks to Theorem \ref{theo:rand}, it suffices to verify that for the randomized realizations
\beq
u(x, \tilde \M^1) = u(x, \tilde \M^2).
\eeq
We introduce randomized model realizations $\tilde \M^1$ and $\tilde \M^2$.
Let $B^1 \in \cA(x, {\bf M}^1)$. Recall that the paths of $B^1$ belong to $D([0, T], \bbr^d)$. Due to Lemma 31 of \cite{Chau-Rasonyi2017}, the strategy can be represented as a Borel function of $Y^1$ and the r.v. $\xi^1$ as follows: $B^1 = f(Y^1, \xi^1)$. Put $B^2 := f(Y^2, \xi^2)$. Our aim is to verify that $B^2 \in \cA(x,  {\bf \tilde M}^2)$, where ${\bf \tilde M}^2$ is a model with additional randomization. $B^2$ is $K$-decreasing since $\cL(B^1)$ assigns probability one to the set of $K$-decreasing c\`adl\`ag paths $H_K$. It can be easily established that $H_K$ is a Borel set in $D([0,T]; \bbr^d)$. Analogously, $\cL(S, B^1)$ assigns probability one to a set 
\[
\{ (x, b) : \, x \in C^d([0, T]) \times H_K, \, x_t > 0, \quad  1/x \cdot b_t \in K/ x_t \, \forall t \in [0, T]\},
\]
implying that $B^2$ satisfies the admissibility property. 

It remains to establish that $B^2_t$ is $\tilde \cF^2_t$-measurable. To this end, note that
\[
\cL(B^2_t, Y^2_u - Y^2_t) = \cL(B^2_t) \otimes \cL(Y^2_u - Y^2_t)
\]
for any pair $0 \leq t \leq u \leq T$. 
Then $B^2_t$ is $\sigma\{ Y^{2, t}, \,_t \! Y^2, \xi^2 \}$-measurable and independent of $\,_t \! Y^2$. Besides, $\,_t \! Y^2$ is independent of $(Y^{2, t}, \, \xi^2)$. Due to Lemma 29 of \cite{Chau-Rasonyi2017}, $B^2_t$ is $\sigma\{ Y^{2, t}, \xi^2 \}$-measurable. 
As the choice of $B^1$ is arbitrary,
\[
u(x, {\bf M}^1) \leq u(x, {\bf \tilde M}^2).
\]
By virtue of Theorem \ref{theo:rand}, 
\[
u(x, {\bf M}^1) \leq u(x, {\bf M}^2).
\]
By swapping the positions of ${\bf M}^1$ and ${\bf M}^2$, we arrive at the required assertion.
\end{proof}

\section{Main result}
 \label{sec:main_result}
Our aim is to provide sufficient conditions on convergence of a sequence of models $\mu^n$ to $\mu$ to ensure convergence of the corresponding Bellman functions $u^n$ to $u$. Due to Theorem \ref{theo:indep}, one may describe this convergence in terms of selected realizations $\M^n$ and $\M$.

In the sequel the superscript  $n$ will indicate objects  related with the $n$th model or model realization, e.g., for the value function we write    
\beq
u(x,{\bf M}^n) := \sup_{B \in {\cA}(x,{\bf M}^n)}\E [U( V^{x,B}_T, S^n)].
\eeq

\smallskip
{\bf A.1.} 
There are two continuous functions $m_i:[0,1]\to \bbr_+$ satisfying $m_i(0)=0$, $i=1,2$, and an integrable random variable  $\zeta\ge 0$ such that for all $x\in K$ and $\alpha \in (0,1)$
\beq 
\label{Ua}
U((1-\alpha)x,S) \ge  (1-m_1(\alpha))U(x,S)-m_2(\alpha)\zeta .
\eeq

\smallskip

The  power function 
    \[
        U(x,s) = \frac{1}{1 - \gamma} \ell(x)^{1-\gamma}, \qquad \gamma \in (0,1),
    \]
where function $\ell$ is defined in \eqref{eq:liquidation},
satisfies this assumption with functions $1 - m_1(\alpha) = (1 - \alpha)^{1 - \gamma}$,  $m_2(\alpha) = 0$ and r.v. $\zeta = 0$.

\smallskip

{\bf A.2.}
The sequence of probability measures $\cL (Y^n, S^n \, | \, \P^n) \to \cL(Y, S \, | \, \P)$ in the Skorokhod $J_1$ topology on $D([0,T]; \bbr^{m+d})$. 

\smallskip

In this setting, we assume that the solvency cone $K$ is the same for all models. The continuity of the Bellman functions with respect to the transaction costs is left for future research. 

\smallskip
{\bf A.3.}
For every $x \in {\rm int} \, K$, the  sets   of random variables 
\[
\{ U(V^{x, B^n}_T, S^n) \colon n \in \bbn, \, B^n \in \cA(x, {\bf M}^n)  \}, \qquad \{ U( V^{x, B}_T, S) : \,  B \in \cA(x, {\bf M})  \}
\]
are uniformly integrable. 
\smallskip

Recall that the involved random variables can be defined on different probability spaces but the uniform integrability is a property of distributions. 
Verifying the uniform integrability assumption can be particularly cumbersome. In this regard, the following lemma would be helpful.

\begin{lemm}
\label{lemm:ui}
Suppose that $U$ is upper-semicontinuous with respect to $x$ and there exist $C > 0$, $\gamma \in \,]0, 1[$ and $q > 1 / (1 - \gamma)$ satisfying the following:
\begin{enumerate}
\renewcommand{\labelenumi}{(\roman{enumi})}
\item for all $(x, s) \in K \times \bbr^d_+$,
\[
U(x, s) \leq C\left(1 + \wp^\gamma(x)\right),
\]
where function $\wp$ was defined in \eqref{eq:purchase};
\item there exists a consistent price system $Z$ such that
\[
\E[ \left( Z_T^1 \right)^{1-q}] < \infty.
\]
\end{enumerate}
Then {\bf{A.3}} holds true.
\end{lemm}

The proof is given in Appendix. 

\smallskip

A c\`adl\`ag process $Z = (Z_t)_{t \le T}$ is called a strictly {\it consistent price system} if $Z$ is  ${\bf F}$-martingale, $ Z_t / S_t \in {\rm int} \, K^*$ and $Z^{1}_0 = 1$. Let $\varepsilon > 0$;  $Z$ is a $\varepsilon$-{\it uniformly consistent price system} if  $Z$ is consistent price system such  that $ Z_t / S_t \in \varepsilon$-${\rm int} \,  K^*$, where 
\[
\varepsilon\hbox{-}{\rm int} \, K^* = \{ y \in \bbr^d \setminus \{ 0 \} \colon  w y > \varepsilon \, |y|  |w|, \ \forall\, w \in K \}.
\]

\smallskip

{\bf A.4.}
There is  $\varepsilon > 0$ such that for every realization ${\bf M}^n$ there is an $\varepsilon$-uniformly consistent price system $Z^n$ and for the realization ${\bf M}$ there is an $\varepsilon$-uniformly consistent price system $Z$. Furthermore, the sequence of measures $\P^n$ is contiguous with respect to $\Q^n := Z^{n, 1}_T \P^n$. 

\smallskip


The contiguity assumption for convergence problems of financial market models was also considered in \cite{hubalek1998}.

\smallskip

Fix $x \in {\rm int} \, K$. A sequence $B^n \in \cA(x, {\bf M}^n)$ is called asymptotically optimal if
\[
\lim_n \left( u(x, {\bf M}^n) - \E[ U( V^{x, B^n}_T, S^n)]  \right) = 0.
\]

\smallskip

The main results of this paper can be summarized in the following 
\begin{theo}
\label{theo:main}
Let  ${\bf A.1}$ -- ${\bf A.4}$ hold. Then for every $x \in {\rm int }\, K$
\[
\lim_{n \to \infty} u(x, {\bf M}^n) = u(x, {\bf M}).
\]
Suppose that $B^n \in \cA(x, {\bf M}^n)$ is an asymptotically optimal sequence of strategies. Then there is a strategy $B \in \cA(x, \M)$ such that $u(x, {\bf M}) = \E[ U( V^{x, B}_T, S)]$ and a subsequence of laws $\cL(Y^n, S^n, B^n)$ converges to $\cL(Y, S, B)$ in the product of the Skorokhod $J_1$ topology on the first two coordinates and the Meyer--Zheng topology on the third.
\end{theo}

In particular, by taking an identical sequence ${\bf M}^n = {\bf M}$, we get that 
\begin{coro}
Let  ${\bf A.1}$ -- ${\bf A.4}$ hold. Then for every $x \in {\rm int}\, K$ there is $B \in \cA(x, \M)$ such that $u(x, {\bf M}) = \E[ U( V^{x, B}_T, S)]$.
\end{coro}

\section{Lower semi-continuity}
\label{sec:lowe}

In this section, we establish the lower semi-continuity. We take an arbitrary admissible strategy $B$ from the limit model realization ${\bf M}$ and approximate it by a sequence of strategies from ${\bf M}^n$. Before we proceed, we prove the continuity of the Bellman function under our assumptions.

For  $r >0$ we define  the ball $\cO_r(x):=\{y\in \bbr^d\colon |y-x|< r\}$ with the closure $\bar \cO_r(x)$. Recall that the liquidation function 
\begin{equation} \label{eq:liquidation}
\begin{split}
    x\mapsto \ell(x)
    &:=\sup\{ \lambda\in \bbr\colon x-\lambda e_1\in K\}\\
    &=\sup\{\lambda\in \bbr\colon \exists\; a\in -K \hbox{ such that } x+a=\lambda e_1\}
\end{split}
\end{equation}
is continuous. Since  $\ell(x)>0$ for $x\in 
 {\rm int}\, K$, for such $x$ the ``liquidation'' strategy $L^x$ with 
  $\Delta L^x_0:=L^x_{0}-L^x_{0-}=\ell(x)e_1-x$ and $L^x_t\equiv \ell(x)e_1$ for $t\ge 0$ is in $\cA(x)$ and $V^{x,L^x}_T\in {\rm int}\, K$.  

\smallskip
It is well-known that a proper concave function  is continuous (and even locally Lipschitz) on its domain. 
Under the assumption {\bf A.1}  we have a stronger property.

\begin{lemm}
\label{lemma1}
Suppose that {\bf A.1} holds. Then $u$ is continuous on ${\rm int}\,K$. 
\end{lemm}
\begin{proof}
Let $x\in {\rm int}\,K$ and let   $B\in \cA(x)$.  
For all sufficiently small $\e>0$ the ball  $\bar \cO_\e(x)=x+\bar \cO_\e(0) \subset {\rm int}\, K$ so that 
 $\ell(y)e_1\in {\rm int}\, K$ for all $y\in \bar\cO_\e(x) $.   
  For any $\alpha\in ]0,1]$ the strategy $B^{\alpha,x}:=\alpha L^x+(1-\alpha)B$ belongs to $\cA(x)$ and $x+(1/S)\cdot B^{\alpha,x}\in \alpha \ell(x)+\widehat K$. It is easily seen 
that
\[
\begin{aligned}
y+(1/S)\cdot L^{y-x} +(1/S)\cdot B^{\alpha,x}
&=\ell (y-x)e_1+x+(1/S)\cdot B^{\alpha,x}\\
&\in (\alpha \ell(x) +\ell (y-x))e_1+\widehat K
\end{aligned}
\]
and one can find $\e_0=\e_0(\alpha)>0$ such that $\alpha \ell(x) +\ell (y-x)\ge 0$ 
for all $y\in \bar \cO_{\e_0}(x)$. 
Then the strategy
$L^{y-x} + B^{\alpha,x}\in \cA(y)$ and 
\[
y+(1/S)\cdot (L^{y-x} +B^{\alpha,x})_T=(\alpha \ell(x) +\ell (y-x))e_1+(1-\alpha)V^{x,B}_T\ge (1-\alpha)V^{x,B}_T.
\]
Due to monotonicity of $U$ and the bound \eqref{Ua} we have that 
\[
\begin{aligned}
&U(y+(1/S)\cdot (L^{y-x} +B^{\alpha,x})_T,S)\\
&\quad \ge U((1-\alpha)V^{x,B}_T,S)\\
&\quad \ge (1-m_1(\alpha))U(V^{x,B}_T,S)-m_2(\alpha)\zeta
\end{aligned}
\]
and, hence, 
\[
u(y)\ge (1-m_1(\alpha))u(x)  - m_2(\alpha)\E [\zeta]. 
\]
It follows that 
\[
\lim_{\e\to 0}\inf_{y\in  \bar \cO_\e(x) }u(y)\ge (1-m_1(\alpha))u(x)  - m_2(\alpha)\E [\zeta].  
\]
Letting $\alpha\downarrow 0$ we get the bound
\beq
\label{inf}
\lim_{\e\to 0}\inf_{y\in  \bar \cO_\e(x) }u(y)\ge u(x). 
\eeq

\smallskip
On the other hand, if $x+\bar \cO_\e(0)\in {\rm int}\, K$, then there is $\kappa=\kappa(\e)>1$ such that 
\[
x+\frac 1{\kappa-1}\bar \cO_\e(0)\in {\rm int}\, K
\]
or, equivalently, $\kappa x-(x+\bar \cO_\e(0))\in {\rm int}\, K$. Thus, $\kappa x\ge y$ for any $y\in \bar \cO_\e(x)$ and, therefore, 
\[
u(y)\le u(\kappa x):=\sup_{B\in\cA(\kappa x)}\E [U(V^{\kappa x,B}_T,S)]
=\sup_{B\in \cA(x)}\E [U(V^{\kappa x,\kappa B}_T,S)].
\]

Let $\alpha=1-1/\kappa$. We may assume that $\kappa (\e)\downarrow 1$ as $\e \downarrow 0$ implying that $\alpha:=\alpha (\e)\downarrow 0$.  
According to {\bf A.1}
\[
\E [U(V^{\kappa x,\kappa B}_T,S)] \le \frac 1{1-m_1(\alpha)}\E [U(V^{ x, B}_T,S)] +\frac {m_2(\alpha)}{1-m_1(\alpha)}\E[\zeta],
\] 
implying that 
\[
u(\kappa x)\le \frac 1{1-m_1(\alpha)}u(x)+\frac {m_2(\alpha)}{1-m_1(\alpha)}\E[\zeta]. 
\] 
It follows that 
\beq
\label{sup}
\lim_{\e\to 0}\sup_{y\in  \bar \cO_\e(x) }u(y)\le u(x) 
\eeq
and the continuity of $u$ follows from (\ref{inf}) and (\ref{sup}).

\end{proof}

For the uniform norm and modulus of continuity of a function $f:[0,T]\to \bbr$ we use the notations 
\[
||f||=||f||_T:=\sup_{s\le T}|f_s|,\qquad  w(f,\e)=w_{T}(f,\e):=  \max_{|r-s|\le \e}|f_r-f_s|. 
\]
We recall an elementary lemma on  approximation of the distribution function of a finite signed measure on $[0,T]$, see \cite{Protter}, Lemma 12.3.  

\begin{lemm} 
\label{lemma2}
Let $b=(b_t)_{t\in [0,T]}$ be a c\`adl\`ag function of bounded variation and let 
\begin{equation}
\label{appr1}
\begin{split}
b^m&:=b_0I_{[t_0,t_1[}+\sum_{k = 1}^{m-1} (b_{t_k}-b_{t_{k-1}})I_{\Delta_k}+b_TI_{\{T\}},\\
\Delta_k&:=[t_k,t_{k+1}[, \quad  t_k=t_k^m:=kT/m.
\end{split}
\end{equation}
Then for any continuous  function $f:[0,T]\to \R$  and any $t_k$ (in particular, for $t_m=T$)
\beq
\label{ineq1+}
|f\cdot b^m_{t_k}- f\cdot b_{t_k}|\le w_{t_k}(f,T/m){\rm Var}_{t_k} b \le  w(f,T/m){\rm Var}_T b. 
\eeq
\end{lemm} 

\smallskip
{\sl Remark.} The function $b$ is the distribution function of a  measure $b(ds)$, while the function $b^m$ is the distribution function of the discrete  measure 
\[
b^m(ds):=b_0\delta_0(ds)+\sum_{k = 1}^{m} (b_{t_k}-b_{t_{k-1}})\delta_{t_k}(ds). 
\]
Lemma \ref{lemma2} implies the weak$*$ convergence of the sequence $b^m(ds)$ to $b(ds)$.  

\smallskip 

\smallskip
Let $B=(B_t)_{t\in [0,T]}$ be a $d$-dimensional adapted c\`adl\`ag process of bounded variation
such that $\dot B_t\in -K$ for all $t\in [0,T]$. Define its piecewise-constant approximation $B^m$ as in   (\ref{appr1}). Then $B^m$ is a strategy, that is  $\dot B^m_t\in -K$ for all $t\in [0,T]$. Indeed, 
\[
\dot B^m_{t_k}=\Delta B^m_{t_k}=B_{t_k}-B_{t_{k-1}}=\int_{[t_{k-1},t_{k}[}\dot B_s d{\rm Var}B_s\in -K 
\]
and, hence, $\dot B^m_{t}\in -K$ for all $t\in [0,T]$.  

Note that in the above arguments it is important that the cone $K$ is constant. 

\begin{lemm} 
\label{lemma3}
Let $\e\in ]0,1]$ be such that $\cO_{\e}(x)\subset {\rm int}\,K$. Let 
$B\in  \cA(x)$ and let $B^m$ be the strategy defined as in  (\ref{appr1}). 
Then 

$(i)$ $|\widehat V^{x,B^m}_T-\widehat V^{x,B}_T|\to 0$ as $m\to \infty$; 

$(ii)$
 there is an  $\bbr^d_+$-valued adapted c\`adl\`ag process $\xi^m=\xi^m(x,S,B)$ with jumps only at the points $t_k$ and such that $||\xi^m||\to 0$ as $m\to\infty$ and \[
\widehat V^{x,B^m}_t+\xi^m_t\in \widehat K_t \quad \hbox{for all $t\in [0,T]$}.  
\]
\end{lemm}
\begin{proof}
Recall that  
$
\widehat V^{x,B,i}
:=x^i+(1/S^i)\cdot B^i$ and $(i)$ follows directly from (\ref{ineq1+}) since $w(1/S^i,T/m)\to 0$ a.s. 
The property $B\in  \cA(x)$ means that $\widehat V^{x,B}_t\in \widehat K_t$
for all $t\in [0,T]$. 
According to (\ref{ineq1+}) 
\[
(1/S^i)\cdot B^{m,i}_{t_k}+\tilde \xi^{m,i}_{t_k}\ge (1/S^i)\cdot B^{i}_{t_k} ,  
\]
where $\tilde \xi^{m,i}_{t_k}:=w_{t_k}(1/S^i,T/m){\rm Var}_{t_k} B^i\ge 0$.  It follows that 
$
\widehat V^{x,B^m}_{t_k}+\tilde \xi^m_{t_k}\in \widehat K_{t_k} 
$
for all $k\le m$. The obvious identity  $
\varphi_{t}\varphi_{t_k}^{-1} \widehat K_{t_k}=\widehat K_t 
$ implies that  
\[
\varphi_{t}\varphi_{t_k}^{-1}(\widehat V^{x,B^m}_{t_k}+\tilde \xi^m_{t_k})\in \widehat K_{t}.
\]
For $t\in [t_k,t_{k+1}[$ we have $\widehat V^{x,B^m}_{t}=\widehat V^{x,B^m}_{t_k}$ and  $\widehat V^{x,B^m}_{t}+\xi^m_{t}\in \widehat K_{t}$, where 
\[
\xi^m_{t}:=(\varphi_{t}\varphi_{t_k}^{-1}-I)\widehat V^{x,B^m}_{t_k}+\varphi_{t}\varphi_{t_k}^{-1} \tilde \xi^m_{t_k} 
\] 
and  $I$ is the identity matrix. 

Note that on  $[t_k,t_{k+1}[$ we have the bound 
$|S^i_{t_k}/S^i_t-1|\le ||1/S^i||w(S^i,T/m)$. Also,  $S^i_{t_k}/S^i_t\le ||S^i||||1/S^i||$. 
In virtue of (\ref{ineq1+})
\[
\begin{aligned}
||(1/S^i)\cdot B^{m,i}||
&\le  ||(1/S^i)\cdot B^{i}||+\tilde \xi^{m,i}\\
&\le ||(1/S^i)||{\rm Var}_{T} B^i
+w(S^i,T/m){\rm Var}_{T} B^i.
\end{aligned}
\]
Summarizing, we get that 
\[
\begin{aligned}
||\xi^{m,i}||
&\le w(S^i,T/m)||1/S^i||\Big[|x^i|\\
&\quad+\big(||(1/S^i)||+ w(S^i,T/m)+||S^i||\big){\rm Var}_{T} B^i \Big]
\end{aligned}
\]
and the result follows. 
\end{proof}

The next two lemmata on approximation are of general interest and give extensions of classical results to the 
case of the cone-valued random variables. 

Recall that $\rho(\eta_1,\eta_2):=\E [|\eta_1-\eta_2|\wedge 1]$ is a metric on 
$L^0(\bbr^d,\P)$ defining the convergence in probability.   
\begin{lemm} 
\label{lemma4} 
Let $\delta>0$, let $\zeta$ be an $\bbr^k$-valued random variable, and let $\eta$ be a 
$\sigma\{\zeta\}$-measurable  random variable  with values  in a closed convex cone $G$. Then there exists 
a bounded continuous function $f:\bbr^k\to G$ such that $\E[|f(\zeta)-\eta |\wedge 1]<\delta$. 
\end{lemm}
\begin{proof}
Put $\eta_c:=\eta I_{\{|\eta|\le c\}}+c\eta /|\eta|I_{\{|\eta| > c \}}$ where the constant $c>0$. The r.v.  $\eta_c$ is $\sigma\{\zeta\}$-measurable, takes values in $G$, 
$|\eta_c|\le c$, and  $\E[|\eta - \eta_c |\wedge 1]<\delta/2$ when $c$ is large enough. 
By the Doob theorem $\eta_c=f_0(\zeta)$ for some Borel function $f_0$ with $|f_0|\le c$. 
The set of real-valued bounded continuous functions is dense in $L^2(\bbr^k,\cB(\bbr^k),\mu_\zeta)$, where $\mu_\zeta$ is the distribution $\zeta$, so that  
there is a  continuous function $f_1$ with $|f_1|\le c$ and  such that 
$\E[|f_1(\zeta)-f_0(\zeta)|^2]< \delta^2/4$. Put $f(y):=\Pi (f_1(y))$, where $\Pi$ is the Euclidean projection onto the closed convex cone $G$. Recall that 
$\Pi$ is a continuous mapping and 
$|\Pi u-v|^2\le |u-v|^2$ for any $v\in G$ and $u\in \bbr^k$.  Thus
\[
 \E[|f(\zeta)-\eta_c |\wedge 1]\le (\E[|f(\zeta)-f_0(\zeta)|^2])^{1/2}<\delta/2 
\]
implying the result. 
\end{proof}

In the following lemma the $\sigma$-algebra $\cF_t$ is generated by a $p$-dimensional   c\`adl\`ag  process $Y=(Y_s)_{s\le t}$, i.e. $\cF_t=\sigma\{Y_s,\; s\le t\}$.

\begin{lemm} 
\label{lemma4+} 
Let $\delta>0$ and let $\eta$ be an $\cF_{t}$-measurable  random variable  with values in a closed convex cone $G$. Then there are $r_i\in [0,t]$,  $i=1,...,M$, and a bounded continuous function 
$f: (\bbr^{p})^M\to G$ such that $\E [|f(Y_{r_1},...,Y_{r_M})-\eta|\wedge 1]<\delta$. 
\end{lemm}
\begin{proof}
As in the above proof, we can reduce the problem to the approximation of the random variable $\eta_c$ such that  $\E[|\eta - \eta_c |\wedge 1]<\delta/2$. 
Let $r_i^N:=i2^{-N}t$, $i=0,...,2^N $, $n\in \N$, and let  
$\cG^N_t:=\sigma\{Y_{r_i^N},\  i\le 2^N \}$.  
Then $(\cG^N_t)_{N\ge 1}$ is a discrete-time filtration with $\sigma\{\cG^N_t,\; N\ge 1\} =\cF_{t}$. 
By the L\'evy theorem 
$\eta_c^N:=\E [\eta_c|\cG^N_t]\to \eta_c$ almost surely as $N\to \infty$. Hence,  
$\E[|\eta_c^N - \eta_c |\wedge 1]<\delta/4$ for sufficiently large $N$. It remains to apply the previous lemma to $\eta_c^N$. 
\end{proof}

Lemma \ref{lemma4+}  implies as an obvious corollary  the following 
\begin{lemm} 
\label{lemma4++} 
 Let $\mathcal{B}_m({\bf F})$ be the set of ${\bf F}$-adapted  processes constant on the intervals $[t_k,t_{k+1}[$, $t_k=
kT/m$, with jumps $\Delta B^m_{t_k}\in L^0(-K,\cF_{t_k})$ and let $\mathcal{B}^c_m({\bf F})$ be its subset consisting of the processes  such that  $\Delta B^m_{t_k}$ can be represented  as bounded continuous functions of values of the process $Y$ at a finite number of points on $[0,t_k]$. Then  for any process $B^m\in \mathcal{B}_m({\bf F})$ there exists a sequence of processes 
$B^{m,l}\in \mathcal{B}^c_m({\bf F})$ such that $||B^{m,l}-B^m||\to 0$ a.s. 
as $l\to \infty$. 
\end{lemm}



\begin{prop}
\label{prop:1}
Suppose that {\bf A.1},  {\bf A.2}, and {\bf A.3} hold. Let $x\in {\rm int}\,K$. Then  
\[
u(x) \le \liminf_n u^n(x). 
\]
\end{prop}
\begin{proof}
Using the Skorokhod theorem we can construct model realizations ${\bf M}$ and ${\bf M}^n$ having the same underlying probability space and such that $||S-S^n||\to 0$ (recall that convergence of continuous paths in $J_1$ topology implies uniform convergence).

Let $\e\in]0,1]$ be such that $\cO_{\e}(x)\in {\rm int}\,K$. Then $x_{\e}:=x - (\e/\sqrt{2d}) \1\in \cO_{\e}(x)$. 
Take  an arbitrary strategy 
$B\in  \cA(x_{\e},{\bf M})$.  We show the inequality
\beq
\E \left[U\left(V^{x_{\e},B}_T, S\right)\right] \le \lim_{n} u^n(x), 
\eeq
implying that $u(x_{\e},{\bf M})\le \lim_n u^n(x)$. Due to Lemma \ref{lemma1} on the continuity of $u$, this  leads to  the needed  assertion.

By Lemma \ref{lemma3}, applied with $x_\e$ and $B\in  \cA(x_{\e},{\bf M})$, there are a piecewise constant  strategy $B^m$ taking values 
$B^m_{t_k}=B_{t_k}$ on the interval $[t_k,t_{k+1}[$ of  length $T/m$ and an adapted 
$\bbr^d_+$-valued c\`adl\`ag process $\xi^m$, jumping only at the points $t_k$,  such that 
$\widehat V^{x_\e,B^m}_t+\xi^m_t\in \widehat K_t=K/S_t$ for all $t\in [0,T]$,  $||\xi^m||\to 0$ a.s., and     
\beq
\label{vit}
|\widehat V^{x_\e,B^m}_T-\widehat V^{x_\e,B}_T|\to 0, \ \ \hbox{a.s.}\quad m\to \infty.
\eeq
%

Since 
\[
\wh V^{ x_{2 \e/3 },B^m}_t = \wh V^{ x_{\e},B^m}_t +x_{2 \e/3 }-x_{\e}=\wh V^{ x_{\e},B^m}_t+2\e'\1, \quad \e':=\frac {\e}{6\sqrt{2d}},
\]
we have that  $\wh V^{ x_{2 \e/3 },B^m}_t -\e'\1\in \widehat K_t$ for all $t\in [0,T]$ on the set $\{\max_i ||\xi^{m,i}||\le 
\e'\}$. 

Thus, for every $\omega$ except a $\P$-null set $\Omega_0$ there is $m_0(\omega)$ such that for all $m\ge m_0(\omega)$ 
\beq
\label{eq:29}
\widehat V^{ x_{2 \e/3 },B^m}_t -\e' \1 \in \widehat K_t\ \forall t \in [0,T].
\eeq 

According to  Lemma \ref{lemma4++}, for each $m\ge 1$ there is a sequence of processes $B^{m,l} \in \mathcal{B}^c_m({\bf F})$, such  that $\lim_l ||B^{m,l}-B^m||\to 0$ except $\P$-null set $\Omega_m$. 

It follows that on $\Omega_m$  
\beq
\lim _l || \wh V^{x_{2\epsilon/3}, B^{m,l}} - \wh V^{x_{2\epsilon/3}, B^{m}}||=\lim_l||(1/S)\cdot B^{m,l}-(1/S)\cdot B^{m}||= 0.  
\eeq
The involved processes being piecewise  constant,  this convergence is just a convergence of random variables. 

By a standard diagonal argument, there is a sequence $B^{m(l(n)),l(n),n}$ of piecewise constant $\F^n$-adapted processes, denoted $C^n$, with the following properties: 
\beq
\label{2delta}
\lim_n \P\left[ \wh V_t^{x_{2\e/3}, C^n} -\e'\1 \in \wh K_t^n\  \forall\, t \in [0,T] \right]=1,
\eeq

\beq
\label{CnB}
\lim_n|\wh V^{C^n}_T - \wh V^{B}_T |= 0\quad \hbox{a.s.}
\eeq

\smallskip

Although $\dot C^n \in -K$, the admissibility property $\wh V^{x_{2\e/3},C^n}_t \in \wh K_t^n$ for all $t\le T$ may not hold. To get an admissible strategy we introduce a ${\bf F}^n$-stopping time
\beq
\label{tau}
\tau^n :=  \inf\Big\{t\colon   \wh V^{x_{2\e/3},C^n}_t - \e'\1\notin \wh K^n_t \Big\}= \inf\Big\{t\colon   \wh V^{x_{2\e/3},C^n}_t - \e'\1 \notin K/S^n_t \Big\},
\eeq
and consider the strategy $D^n$ 
\beq
D^n := C^n I_{[0,\tau^n[} + \Big (C^n_{\tau^n-} +  \ell \big(V^{x_{2\e/3}, C^n}_{\tau^n -}\big) e_1 -V^{x_{2\e/3}, C^n}_{\tau^n -}\Big )I_{[\tau^n,\infty[},
\eeq
where $\ell$ is the liquidation function. Note that 
$\dot D^n \in -K$. Indeed, $D^n$ is a piecewise constant 
process coinciding with $C^n$ on $[0,\tau^n[$,  not evolving  on $[\tau^n,T]$, and its jump at  the point $\tau^n$ is equal to  $\ell \big(V^{x_{2\e/3}, C^n}_{\tau^n -}\big) e_1 -V^{x_{2\e/3}, C^n}_{\tau^n -}$ and takes values in  $-K$. 

By definition of $\tau^n$ we have that $\wh V^{x_{2\e/3},D^n}_t\in \wh K_t^n=K/S^n_t$.  
Thus, $D^n\in \cA^{n}(x_{2\e/3})$.   

From  (\ref{tau}) 
\beq
\P\left[\tau^n=\infty \right]=\P\left[ \wh V_t^{x_{2\e/3}, C^n} -\e'\1 \in \wh K_t^n\  \forall\, t \in [0,T] \right].
\eeq

Note also that
\begin{align*}
 u^n(x) &\ge \E\left[ U(V^{x_{2\e/3},D^n }_T, S^n)  \right]\\
  &=
 \E\left[ U( V^{x_{2\e/3},C^n }_T, S^n)I_{\{\tau^n=\infty\}}  \right] +\E\left[ U( V^{x_{2\e/3},D^n }_{\tau^n}, S^n) I_{\{\tau^n\le T\}} \right]
\end{align*}
By virtue of the Fatou lemma, conditions {\bf A.1}, {\bf A.3}, and (\ref{CnB}),  the $\liminf_n$ of the rhs dominates  the value
\[
 \E\left[ U( V^{x_{2\e/3},B}_T, S) \right]\ge  \E\left[ U( V^{x_{\e},B}_T, S) \right]
 \]
and the result holds. 
\end{proof}

\section{Upper semi-continuity}
\label{sec:upper}

This section establishes upper semicontinuity. We will use the following notation. Let $\cM$ be a family of $d$-dimensional right-continuous martingales $Z$ such that $Z_t \in \wh K^*_t$ a.s. for every $t \in [0, T]$. By definition of $\wh K^*_t = \varphi_t^{-1} K^*$ (see Section~\ref{problem}), this is equivalent to $Z_t / S_t \in K^*$ a.s. for every $t \in [0, T]$. Also, denote $\cX(a)$ the set of $d$-dimensional right-continuous processes of bounded variation $X$ with the Radon--Nikodym derivative (with respect to  ${\rm Var} \, X$, the total variation process) $\dot X_t \in - \wh K_t$ a.s. for every $t$ and $X + a {\bf 1} \in \wh K_t$ a.s. for every $t \in [0,T]$. Put $\cX = \bigcup_{a} \cX(a)$. The next lemma is identical to Lemma 3.6.2. from \cite{KabS}.

\begin{lemm}
\label{lemma:4.1}
    Let $X \in \cX$ and $Z \in \cM$. Then the scalar product $ZX$ is a supermartingale and 
    \beq
    \E [(-Z \dot X) \cdot {\rm Var}\, X_T] \leq -\E[Z_T X_T].
    \eeq
\end{lemm}
\begin{proof}
    We will denote the Hadamard (component-wise) product of vectors as $\odot$. By the product formula,
    \[
    Z \odot X = X_- \cdot Z + Z \cdot X = \int_{]0, \cdot]} X_{s-} d Z_s + \int_{[0, \cdot]} Z_s d X_s.
    \]
    By convention, $Z_0 \odot X_0 = Z \cdot X_0$.
    The first term is a local martingale. As $Z_t \dot X_t \leq 0$ a.s., the last term is a negative decreasing process.  Note that
    \[
    (X_- \cdot Z) \1 = ZX - Z \dot X \cdot {\rm Var} \, X  \geq Z X \geq -aZ \1.
    \]
    for some $a \in \bbr_+$ by definition of the set $\cX$. Recall that a local martingale bounded from below by a martingale is a supermartingale. The terminal value of the decreasing process $Z \cdot X_T = Z_T X_T - X_- \cdot Z_T $. Since $X_- \cdot Z$ is a supermartingale, $\E[ X_- \cdot Z_T] \leq  0$. Finally, $Z_T X_T \geq -a Z_T \1$ implying that $(Z \dot X) \cdot {\rm Var } \, X_T$ is bounded from below by an integrable random variable. Thus, $(Z \dot X) \cdot {\rm Var } \, X$ is a supermartingale and so is $ZX$. Immediately,
    \[
    \E [(-Z \dot X) \cdot {\rm Var} X_T] = \E[X_- \cdot Z_T - Z_T X_T] \leq -\E[Z_T X_T].
    \]
    The inequality follows.
\end{proof}

Put $\Lambda := K^* \cap \{ x \in \bbr^d : \, x^1 = 1 \}$.
Similarly to the liquidation function of \cite{Bouchard}, we introduce a function:
\begin{equation} \label{eq:purchase}
    \wp (x) := \inf \{ y \in \bbr: \: ye_1 - x \in K \}.
\end{equation}

Henceforth, we will refer to it as purchase function. This function calculates the minimal amount of the numeraire required to enter a given position $x$. Standard arguments imply the following dual representation
\[
\wp (x) = \sup_{y \in \Lambda} xy.
\]
It is easily seen that $\wp (x) = - \ell (-x)$.

\begin{lemm}
\label{lemma:4.2}
Let  ${\bf A.4}$ be fulfilled.  Put $\Q = Z^1_T \P$. Then 
$\E_Q [{\rm Var}\, B_T] \leq {\wp(x)}/{\varepsilon}$.
\end{lemm}
\begin{proof}
Put $M := Z / Z^1$.
Notice that $-\dot B_t \in K$ and the definition of $Z$ imply that
\[
-(\dot B_t / S_t) M_t  = -\dot B_t (M_t / S_t) \geq \varepsilon |\dot B_t| |M_t / S_t|.
\]
By definition, $|\dot B_t(\omega)|$ = 1 outside a set of measure $d {\rm Var}_T \, B(\omega) d\P(\omega)$ zero. By changing on the set of $d {\rm Var}_T \, B(\omega) d\P(\omega)$ measure zero, we assume that $|\dot B| = 1$. Also, note that $M^1_t = 1$ and $S_t = 1$ implying that $|M_t / S_t| \geq 1$. We obtain that $-\dot B_t (M_t / S_t ) \geq \varepsilon$.

Since $M$ is a $Q$-martingale,  Lemma \ref{lemma:4.1} yields
\[
\E_Q [(-M \cdot \wh V_T ) \1] \leq -\E_Q[ M_T \wh V_T].
\]
Due to the fact that $\wh V_{0-} = x$ and $S_0 = \1$,
\begin{multline*}
\E_Q [(-M \cdot \wh V_T ) \1]  = -x \E_Q [M_T] - \E_Q [M / S \cdot B_T] =  \\
-x \E_Q [M_0] + \E_Q [(-(M/S) \dot B) \cdot {\rm Var} \, B_T] \geq -x \E_Q [M_0] + \varepsilon \E_Q [{\rm Var} \, B_T].
\end{multline*}
Finally,
\[
\varepsilon \E_Q [ {\rm Var} \, B_T] \leq x \E_Q [M_0] - \E_Q [M_T \wh V_T ]\leq x \E_Q[ M_0].
\]
 As $\E [M_0 ]\in \Lambda$, the assertion follows.
\end{proof}

The following statement follows immediately from the contiguity of $\Q^n$, the Markov inequality, and the above lemma.
\begin{coro}
\label{coro:4.3}
Let  ${\bf A.4}$ be  fulfilled. Then for every sequence $B^n \in \cA(x, {\bf M}^n)$ and $\delta > 0$,  there exist $c > 0$ and $N \in \bbn$ such that for all $n > N$ 
\[
\P^n[{\rm Var} \, B^n_T > c] < \delta.  
\]
\end{coro}

Verifying weak convergence of a sequence of probability measures relies on the Prokhorov theorem. In order to invoke this theorem for measures on a space of paths, specifying the topology is essential. For instance, $D([0, T], \bbr^m)$ is equipped with the Skorokhod topology making it a Polish space. However, there is another topology that is suitable when the measure is concentrated in a set of paths of finite variation. In the sequel, we  use the Meyer--Zheng topology, \cite{MZh}. We will denote $D_{MZ}([0, T], \bbr^d)$ or, simply,  $\cD^d_{T, MZ}$ the Skorokhod space endowed with this topology. Following \cite{BDD}, we define the Meyer--Zheng topology by a metric
\beq
d_{MZ}(x, y) = \int_{[0, T[} \min(|x(t) - y(t)|, \, 1) dt + \min(|x(T) - y(T)|, \,1).
\eeq

This metric is simply a metric of convergence in measure ${\rm Leb}([0, T]) + \delta_T$. Note that if $x_n \to x$ in this topology, then, thanks to the Riesz theorem, there is such a subsequence $n_k$ that $x_{n_k} \to x$ a.e. Since both $x_{n_k}$ and $x$ are right-continuous, $x_{n_k} \to x$ pointwise in $[0, T]$.  Finally, it is worth mentioning that, due to the Helly theorem (see, e.g., Theorem 12.7 of \cite{Protter}), the set $\{ x \in \cD^d_{MZ} : \, {\rm Var} \, x \leq c \}$ is compact in the Meyer--Zheng topology for any $c > 0$.

\smallskip
To avoid potential ambiguity arising from the selection of the probability space, we employ a purely measure-theoretic notation in the subsequent lemmas. 

\begin{lemm}
\label{lemma:4.4}
Let  $\bf{A.2}$ and ${\bf A.4}$ hold. 
Let $(S^n, Y^n)$ and $(S, Y)$ be the processes from the models ${\bf M}^n$ and ${\bf M}$, respectively. Let $B^n \in \cA(x, {\bf M}^n)$. Then the sequence of measures $\cL(S^n, Y^n, B^n \, | \, \P^n)$ on the space $C^d([0, T]) \times D([0, T], \bbr^m) \times D_{MZ} ([0, T], \bbr^d)$ is relatively compact. Every cluster point $\nu$ has the following property: the marginal $\nu \circ \pi^{-1}_{C^d([0, T]) \times D([0, T], \bbr^m)} = \mu = \cL(S, Y| \, \P)$ and $\nu \circ \pi^{-1}_{D_{MZ} ([0, T], \bbr^d)}(H_{K}) = 1$ where $H_{K}$ is a family of all $K$-decreasing paths in $D_{MZ} ([0, T], \bbr^d)$.
\end{lemm}
\begin{proof}
For any $c > 0$, the set $\{ x \in D_{MZ} ([0, T], \bbr^d) \colon  {\rm Var}_T \, x \leq c \}$ is compact in the Meyer--Zheng topology. Corollary \ref{coro:4.3} and the Prokhorov theorem imply that $\cL(S^n, Y^n, B^n | \P^n)$ is relatively compact, so the existence of a cluster point $\nu$ is established. For the sake of brevity, we denote a subsequence converging to $\nu$ again by $\cL(S^n, Y^n, B^n | \P^n)$. 

Now it remains to establish the latter assertion. Note that $H_K$ is closed in the Meyer--Zheng topology. Indeed, let paths $x^n \to x$ converge in this topology. It means that $x^n \to x$ in measure ${\rm Leb}([0, T]) + \delta_T$. Due to the Riesz theorem, $x^n \to x$ a.e. up to a subsequence. As both $x^n$ and $x$ are right-continuous, $x^n(t) \to x(t)$ for every $t \in [0, T]$. It remains to highlight that the partial ordering $\preceq
_K$ is closed. 

Since $H_K$ is closed, the weak convergence gives that 
\[
1 = \limsup_n \P^n[B^n \in H_K] \leq \nu \circ \pi^{-1}_{D_{MZ} ([0, T], \bbr^d)}(H_{K})
\] 
and the last statement of the lemma follows. 
\end{proof}

The following lemma establishes the admissibility property of the cluster point. 
\begin{lemm}
\label{lemma:4.5}
Let $x \in {\rm int} \, K$ and
\[
\begin{aligned}
G_{x, K} := \{(z, y, b) \in  C^d([0, T]) \times D([0, T], \bbr^m) \times H_K:\;& z_t > 0,\\
& x + 1/z \cdot b_t \in K / z_t \; \forall t \in [0, T] \}.
\end{aligned}
\]
Then under the hypothesis of Lemma \ref{lemma:4.4}, $\nu(G_{x, K})$ = 1.
\end{lemm}
\begin{proof}
It suffices to verify that $G_{x, K}$ is a closed set. Denote a sequence $(z^n, y^n, b^n) \in G_{x, K}$, $n \in \bbn$, a point $(z, y, b) \in G_{x, K}$ and suppose that $(z^n, y^n, b^n) \to (z, y, b)$ in the product topology. As $1/z^n \to 1/z$ uniformly and $b^n \to b$ point-wise, $1/z^n \cdot b^n \to 1/z \cdot b$ pointwise by virtue of Theorem 12.16  from the book \cite{Protter}. Due to the right-continuity,  
\[
1 = \limsup_n \P^n[(S^n, Y^n, B^n) \in G_{x, K}] \leq \nu (G_{x, K})
\]
and the assertion follows.
\end{proof}

\subsection{Completion of the proof of Theorem \ref{theo:main}}

Let $(B^n)$ be an asymptotically optimal sequence.
We select such a subsequence $(S^n, Y^n, B^n)$ that $\cL(S^{n}, Y^n, B^{n}) \to \cL(S, Y, B)$. Then, by selecting a further subsequence, we establish that $u(x, {\bf M}^n)$ converges. 

Now we will use the Skorokhod representation. At first, we invoke Theorem \ref{theo:rand} and consider the models with randomized strategies ${\bf \tilde M}^n$.
By the Skorokhod representation theorem for laws on separable metric spaces (\cite[Theorem~3]{Dudley1968}), applied to the product space $C^d([0,T]) \times D([0,T], \bbr^m) \times D_{MZ}([0,T], \bbr^d)$, we may redefine $(S^n, Y^n, B^n)$ and $(S, Y, B)$ on a common probability space $(\bar \Omega, \bar \cF, \bar \P)$ such that $(S^n, Y^n, B^n) \to (S, Y, B)$ a.s. along a subsequence. Furthermore, $\bar \Omega$ carries uniform r.v. $\xi^n$ independent of $Y^n$ and a uniform r.v. $\xi$ independent of $Y$ and $\xi^n \to \xi$ a.s. along a subsequence.   In this case, $1/S^n \cdot B^n \to 1/S \cdot B $ a.s. up to a subsequence. Indeed, $S^n \to S$ a.s. uniformly and $S > 0$ a.s. and $B^n \to B$ in the Meyer--Zheng topology, and, hence, point-wise along a further subsequence. 

It remains to verify that $B$ is ${\bf F}$-measurable. This argument follows verbatim from the proof of Theorem \ref{theo:indep}. By virtue of, again, Theorem \ref{theo:indep}, the Bellman function remains unaffected by the Skorokhod representation. 

Finally, {\bf A.2}, {\bf A.3}  and the Fatou lemma yield
\[
\begin{aligned}
\lim_n u(x, {\bf M}^n)
&= \lim_n \E [U(  V^{x, B^n}_T, S^n)]\\
&\leq \E[ \limsup_n U(  V^{x, B^n}_T, S^n)]\\
&= \E [U( V^{x, B}_T, S)] \leq u(x, {\bf M}).
\end{aligned}
\]
Together with Proposition~\ref{prop:1}, this shows that the limit $B \in \cA(x, \tilde {\bf M})$ attains the value $u(x, {\bf M}) = u(x, \tilde {\bf M})$. By Remark~\ref{rem:random_max}, the strategy $\bar B(a)$ obtained by fixing $a \in [0, 1]$ from a full-measure set lies in $\cA(x, {\bf M})$ and attains $u(x, {\bf M}) = \E[U(V^{x, \bar B(a)}_T, S)]$. That completes the proof.

\section{Appendix}
\label{sec:appendix}

\subsection{An auxiliary lemma}
 
\begin{lemm}
\label{lemm:var}
Let $X = (X_t)$ be an adapted c\`adl\`ag process and $K$ be a closed proper convex cone. Then the following properties are equivalent:

$(i)$ $X$ is $K$-decreasing;

$(ii)$ $X$ is of bounded variation with the Radon--Nikodym derivative $\dot X := d X / d {\rm Var}\, X$ evolving in $-K$ where 
${\rm Var} \, X_t:=\sum_{i\le d} {\rm Var} \, X^i_t$ is the total variation of $X$ on $[0, t]$. 
\end{lemm}
\begin{proof}
The implication $(ii) \Rightarrow (i)$ is obvious. Let us verify that $(i) \Rightarrow (ii)$. 
As $K$ is proper, $K^*$ has a non-empty interior and, therefore, there is a point $y\in K^*$ such that  $y+e^j \in K^*$ for all vectors $e^j$ of the canonical basis.  It follows that $K^*$ contains  $d$ linear independent vectors $a^j$. 
Define the scalar processes  $Z^j: = a^j X $.
Since  $X_t - X_s \in -K$ for $ s \leq t $, we have $Z^j_t - Z^j_s \le 0$. As $Z^j$ have c\`adl\`ag paths, $Z^j$ are decreasing.  Then the coefficients $C^j_t$,  of a linear combination
\[
X_t = \sum_{j=1}^d a^j C_t^j
\]
can be obtained as linear combinations of scalar products $Z_t$:
\[
C_t = G^{-1} Z_t,
\]
where $G$ is the Gram matrix. It implies that $C^j$ are c\`adl\`ag processes of bounded variation and so is $X$. It follows that ${\rm Var} \, X$ is well-defined by Lemma I.3.3 of \cite{JS}. 

It remains to show that $\dot X \in -K$ a.e.\ with respect to $dP \otimes d{\rm Var}\, X$. Pick a countable dense subset $\{a_k\}_{k \ge 1} \subset K^*$. For each $k$, we have that
\[
a_k(X_t - X_s) = \int_{]s, t]} a_k \dot X_u \, d {\rm Var}\, X_u \le 0
\quad \text{for all } s \le t,
\]
hence $a_k \dot X \le 0$ outside a $dP \otimes d{\rm Var}\, X$-null set $N_k$. The union $N := \bigcup_k N_k$ is null, and outside $N$, $a_k \dot X \le 0$ for every $k \in \bbn$. By density of $\{a_k\}$ in $K^*$ and continuity of the map $a \mapsto a\, \dot X(\omega, t)$, we obtain $a\, \dot X \le 0$ for every $a \in K^*$ outside $N$. By the bipolar theorem, $K = K^{**}$, so $\dot X \in -K$ outside $N$. Redefining $\dot X$ as $0$ on $N$ gives the required version.
\end{proof}

Let $B$ be an $\bbr^d$-valued  ${\bf F}$-adapted c\`adl\`ag $K$-decreasing process. Due to the above lemma, $B$ is a process of bounded variation. 
Denoting by  ${\rm Var}\,B$  the increasing ${\bf F}$-adapted 
c\`adl\`ag process which is the  sum  of total variation processes ${\rm Var}\,B^i$, we put $\dot B_t:=dB_t/d{\rm Var}_tB$ an optional version of the Radon--Nikodym derivatives that exist due to Lemma I.3.13 of \cite{JS}. Again, the above line proves that $\dot B \in -K$ up to evanescence.

\subsection{Proof of Lemma \ref{lemm:ui}}

Put $p := q / (q - 1)$. Then $1/p > \gamma$. It suffices to establish that, for a fixed $x \in \operatorname{int} K$,
\[
\sup_{B \in \cA(x, \mathbf{M})} \E \left[\wp\!\left(V^{x, B}_T\right)\right]^{1/p} < \infty.
\]
Indeed, if $\eta := (1/p)/\gamma > 1$, then
\[
U^\eta(x,s) \le C^\eta 2^{\eta-1}\left(1+\wp^{1/p}(x)\right),
\]
so de la Vall\'ee-Poussin's criterion implies the uniform integrability.

Put $\Q := Z^1_T \P$ and $M := Z / Z^1$. By Lemma \ref{lemma:4.1}, $M\cdot\widehat V^{x, B}$ is a $\Q$-supermartingale. Since $M_T/S_T \in K^*$ and $(M_T/S_T)^1 = 1$, we have $M_T/S_T \in \Lambda$, hence
\[
\wp\!\left(V^{x, B}_T\right) \le \left({M_T}/{S_T}\right)V^{x, B}_T = M_T\widehat V^{x, B}_T.
\]
By H\"older's inequality,
\begin{multline*}
\E \left[ \wp\!\left(V^{x, B}_T\right) \right]^{1/p}
\leq \E \left[M_T\widehat V^{x, B}_T \right]^{1/p}
= \E_Q \left[ \left(M_T\widehat V^{x, B}_T \right)^{1/p} \left(Z^1_T \right)^{-1} \right] \leq \\
\left(\E_Q \left[ M_T\widehat V^{x, B}_T \right] \right)^{1/p}
\left( \E_Q \left[(Z^1_T)^{-q}\right] \right)^{1/q}
\leq \wp^{1/p}(x)\left( \E \left[(Z^1_T)^{1-q}\right] \right)^{1/q} < \infty.
\end{multline*}
This concludes the proof.

\end{document}